\documentclass[12pt,a4paper]{article}

\usepackage[utf8]{inputenc}
\usepackage[T1]{fontenc}
\usepackage{amsmath,amssymb,amsfonts,amsthm}
\usepackage{geometry}
\usepackage{graphicx}
\usepackage{hyperref}
\usepackage{url}
\usepackage{booktabs,multirow,tabularx}
\usepackage{algorithm}
\usepackage{algpseudocode}
\usepackage{tikz}
\usetikzlibrary{shapes.geometric, arrows.meta, positioning, fit, calc, backgrounds}
\usepackage{xcolor}
\usepackage{caption}
\usepackage{float}
\usepackage{listings}

\newcommand{\SHA}{\mathsf{SHA\text{-}256}}

\newcommand{\NVM}{\textnormal{\textsc{n-Vm}}}
\newcommand{\ACEGF}{\textnormal{\textsc{Ace-Gf}}}
\newcommand{\ACERuntime}{\textnormal{\textsc{Ace Runtime}}}

\newtheorem{definition}{Definition}
\newtheorem{theorem}{Theorem}[section]

\newtheorem{proposition}{Proposition}[section]

\title{n-VM: A Multi-VM Layer-1 Architecture with Shared Identity and Token State}

\author{
  \texttt{Jian Sheng Wang} \\
  \texttt{jason@acechain.io} \\
  Montreal, Canada
}
\date{November 24, 2026}

\begin{document}
\maketitle

% ============================================================
\begin{abstract}
Multi-chain ecosystems suffer from fragmented identity, siloed liquidity, and bridge-dependent token transfers.
We present n-VM, a Layer-1 architecture that hosts n heterogeneous virtual machines as co-equal execution environments over shared consensus and shared state.
The design combines three components: a dispatcher that routes transactions by opcode prefix, a unified identity layer in which one 32-byte commitment anchors VM-specific addresses, and a unified token ledger that exposes VM-native interfaces such as ERC-20 and SPL over a common balance store.
We formalize routing, identity derivation, and token transfer semantics, and prove cross-VM transfer atomicity and identity isolation under standard cryptographic assumptions.
We describe a concrete instantiation with five VMs: a native runtime, EVM, SVM, Bitcoin Script, and TVM, the last realized as the EVM engine with opcode remapping.
We also present account-keyed sharding and a read/write-set scheduler for parallel execution.
Under an analytical throughput model with installed-key Ed25519 authorization on a 16-core server, the architecture admits a projected range of about 16,000 transactions per second when every transaction conflicts to a CPU ceiling of about 107,000 when execution parallelizes across all cores.

\end{abstract}

% ============================================================
\section{Introduction}\label{sec:intro}

% ----------------------------------------------------------
\subsection{Motivation}

Major blockchains run on incompatible execution environments.
Ethereum's EVM, Solana's SVM, Bitcoin's Script, and Tron's TVM each define distinct account models, address formats, transaction semantics, and smart contract languages.
Users who wish to operate across these ecosystems must maintain separate wallets, manage distinct key pairs, and rely on cross-chain bridges, third-party systems that have historically been the largest single source of security failures in decentralized finance, with over \$2.8 billion lost to bridge exploits between 2021 and 2024~\cite{bridge-exploits}.

Recent projects have begun exploring multi-VM architectures.
Movement Labs deploys an EVM-compatible execution layer atop a Move-based settlement chain; Eclipse runs an SVM execution environment as an Ethereum rollup; Sei~v2 combines EVM and CosmWasm within a Cosmos SDK chain.
In each of these approaches, however, one VM is subordinate to another (as a rollup or compatibility layer), identity remains fragmented across VMs, and token transfers between execution environments still require bridge-like mechanisms.

% ----------------------------------------------------------
\subsection{Key Insight: VM-Agnostic Identity--Authorization Separation}

The \ACEGF{} framework~\cite{acegf} introduces a separation between \emph{identity binding} and \emph{per-transaction authorization}.
A user's on-chain identity is a single 32-byte value $\texttt{id\_com}$, the ZK-ACE native identifier~\cite{zkace}: a Poseidon2 hash of a root entropy value (REV), a public protocol constant, and a domain. It is not an output of the \ACEGF{} HKDF derivation, which produces the keys.
In \ACERuntime{}~\cite{aceruntime}, an account is born with a zero-knowledge proof of knowledge of its REV, and each transaction then carries one signature (Ed25519 or ML-DSA-44) under a key installed in an on-chain key table for that $\texttt{id\_com}$; the keys can be rotated without changing $\texttt{id\_com}$.

This separation is \emph{VM-agnostic}: the identity commitment and attestation mechanism are independent of any particular execution environment.
A single $\texttt{id\_com}$ can therefore serve as the anchor for addresses across arbitrarily many VMs, provided a deterministic, collision-resistant mapping from $\texttt{id\_com}$ to each VM's native address format exists.

% ----------------------------------------------------------
\subsection{Contributions}

This paper makes the following contributions:

\begin{enumerate}
    \item \textbf{Generalized \NVM{} architecture.} We define a framework for hosting $n$ co-equal virtual machines on a single Layer-1 chain, with opcode-prefix-based transaction routing, a shared state tree, and a pluggable engine interface (Section~\ref{sec:architecture}).
    \item \textbf{Unified identity layer.} We formalize the deterministic derivation of VM-specific addresses from a single identity commitment, prove address namespace isolation with SHA-256 modeled as a random oracle, and describe raw-chain ingress verification for legacy wallets (Section~\ref{sec:identity}).
    \item \textbf{Unified token ledger.} We present a single-ledger token runtime that exposes ERC-20 and SPL interfaces over the same underlying balance storage, and prove that cross-VM token transfers are atomic (Section~\ref{sec:token}).
    \item \textbf{Parallel execution.} We describe a read/write-set conflict detection scheduler and an account-keyed sharding scheme, and analyze throughput scaling (Section~\ref{sec:parallel}).
    \item \textbf{Concrete instantiation.} We describe a Rust implementation with $n=5$ VMs (Native, EVM, SVM, BVM, TVM), including identity precompiles, cross-VM invocation hooks, and raw-chain signature verification (Section~\ref{sec:implementation}).
\end{enumerate}

\NVM{} builds upon the \ACEGF{} identity framework~\cite{acegf} for multi-stream key derivation, the \ACERuntime{} execution layer~\cite{aceruntime} for the Attest--Execute--Prove pipeline, and the VA-DAR recovery protocol~\cite{vadar} for email-anchored wallet recovery.

% ============================================================
\section{Background}\label{sec:background}

% ----------------------------------------------------------
\subsection{The ACE-GF Identity Primitive}

The Atomic Cryptographic Entity Generative Framework (\ACEGF{})~\cite{acegf} derives deterministic, purpose-specific key streams from a single high-entropy secret called the Root Entropy Value (REV).
The $i$-th key stream is:
\begin{equation}
    k_i \leftarrow \text{HKDF-SHA256}(\text{REV},\; \mathit{info}_i,\; \mathit{salt}_i).
\end{equation}

\ACERuntime{}~\cite{aceruntime} uses an application profile of \ACEGF{} context strings for several ecosystems (these strings belong to that profile, not to~\cite{acegf}):
\begin{itemize}
    \item Stream 1: Ed25519 (Solana signing)
    \item Stream 3: secp256k1 (EVM chains)
    \item Stream 4: secp256k1 (Bitcoin)
    \item Stream 7: ML-DSA-44 (post-quantum signing)
\end{itemize}

The on-chain identity is the ZK-ACE native identifier~\cite{zkace}, relation (C1):
\begin{equation}\label{eq:idcom}
    \texttt{id\_com} = \mathsf{S}_1\bigl(\mathsf{Enc}(\text{REV},\; \mathsf{s}_{\mathrm{id}},\; \text{domain})\bigr),
\end{equation}
where $\mathsf{S}_1$ is the Poseidon2-M31 sponge of~\cite{zkace}, $\mathsf{Enc}$ its field encoding, and $\mathsf{s}_{\mathrm{id}}$ a public protocol constant.

This 32-byte value is the primary native identity artifact visible on-chain. It is a one-way image of REV, not a hiding commitment.
The public keys that authorize transactions are installed in the on-chain key table, so an Ed25519 key is exposed to a future quantum adversary; an account can rotate to an ML-DSA-44 key without changing $\texttt{id\_com}$.

% ----------------------------------------------------------
\subsection{The Attest--Execute--Prove Pipeline}

\ACERuntime{}~\cite{aceruntime} processes transactions through a three-phase pipeline:
\begin{enumerate}
    \item \textbf{Attest} (CPU): check the transaction header (${\sim}1$--$5\,\mu$s/tx) and verify its signature under the sender's installed key (${\sim}76\,\mu$s for Ed25519, ${\sim}200\,\mu$s for ML-DSA-44, single core, cited estimates).
    \item \textbf{Execute} (${\sim}10$--$50\,\mu$s/tx): execute the transaction against the state tree.
    \item \textbf{Prove} (off critical path, optional): build an aggregate certificate from public block data that lets outside parties check every authorization of a block with one verification.
\end{enumerate}

Phase~1 and Phase~2 reside on the critical path and determine block time (${\sim}400$\,ms); validators vote after verifying every signature natively, which gives soft finality.
Phase~3 runs asynchronously; its certificate adds hard finality for outside parties and never rolls back a soft-final block.

The \NVM{} dispatcher replaces the monolithic Execute phase with a routing layer that delegates to $n$ pluggable VM engines.

% ----------------------------------------------------------
\subsection{Existing Multi-VM Approaches}

We briefly survey existing multi-VM designs and identify their limitations.

\begin{table}[H]
\centering
\caption{Comparison of multi-VM approaches.}
\label{tab:existing}
\begin{tabularx}{\textwidth}{l>{\raggedright\arraybackslash}X>{\raggedright\arraybackslash}X>{\raggedright\arraybackslash}X}
\toprule
\textbf{System} & \textbf{VMs} & \textbf{Identity} & \textbf{Token Transfer} \\
\midrule
Movement & EVM + Move & Separate per VM & Bridge \\
Eclipse & SVM on Ethereum & Separate per VM & Bridge (L1$\leftrightarrow$L2) \\
Sei v2 & EVM + CosmWasm & Pointer contracts & CosmWasm$\leftrightarrow$EVM pointer \\
Polkadot & Per-parachain & Per-parachain & XCM message passing \\
\addlinespace
\NVM{} (this work) & $n$ co-equal VMs & Single $\texttt{id\_com}$ & Unified ledger (no bridge) \\
\bottomrule
\end{tabularx}
\end{table}

In \NVM{}, all VMs share a single consensus, identity, and token layer, and no VM is subordinate to another.

% ============================================================
\section{Architecture}\label{sec:architecture}

% ----------------------------------------------------------
\subsection{Overview}

The \NVM{} architecture consists of four layers:

\begin{figure}[H]
\centering
\begin{tikzpicture}[
    node distance=0.6cm,
    box/.style={draw, rounded corners, minimum width=12cm, minimum height=0.8cm, align=center, font=\small},
    arrow/.style={-{Stealth[length=5pt]}, thick},
]
    \node[box, fill=yellow!15] (consensus) {Consensus Layer (BFT + Attest--Execute--Prove)};
    \node[box, fill=orange!15, below=of consensus] (dispatcher) {n-VM Dispatcher (opcode routing)};
    \node[box, fill=blue!10, below=of dispatcher, minimum height=1.2cm] (engines) {
        VM$_1$ (Native) \quad VM$_2$ (EVM) \quad VM$_3$ (SVM) \quad $\cdots$ \quad VM$_n$
    };
    \node[box, fill=green!10, below=of engines] (state) {Unified State Tree + Identity Layer + Token Ledger};

    \draw[arrow] (consensus) -- (dispatcher);
    \draw[arrow] (dispatcher) -- (engines);
    \draw[arrow] (engines) -- (state);
\end{tikzpicture}
\caption{Layered \NVM{} architecture.}
\label{fig:architecture}
\end{figure}

% ----------------------------------------------------------
\subsection{Opcode-Based Transaction Routing}

Each transaction carries a payload whose first byte determines the target VM.
We partition the 256-value opcode space into $n$ contiguous ranges:

\begin{definition}[Opcode Routing Function]\label{def:routing}
Let $\mathcal{V} = \{V_1, V_2, \ldots, V_n\}$ be the set of registered VMs.
Each $V_i$ is assigned a contiguous opcode range $[l_i, u_i] \subset [0, 255]$ where the ranges are pairwise disjoint.
The routing function is:
\[
    \mathsf{Route}(\mathit{opcode}) =
    \begin{cases}
        V_i & \text{if } l_i \le \mathit{opcode} \le u_i, \\
        \bot & \text{otherwise (reject transaction).}
    \end{cases}
\]
\end{definition}

This design has several advantages:
\begin{itemize}
    \item \textbf{$O(1)$ routing.} A single byte comparison determines the target VM, adding negligible overhead to transaction processing.
    \item \textbf{Extensibility.} New VMs can be added by registering an engine for an unoccupied opcode range, without modifying existing engines.
    \item \textbf{Determinism.} The routing decision depends solely on the transaction payload, ensuring all validators reach the same dispatch decision.
\end{itemize}

% ----------------------------------------------------------
\subsection{VM Engine Interface}

Each VM engine implements a minimal interface:

\begin{definition}[VM Engine]\label{def:engine}
A VM engine $E_i$ for VM $V_i$ implements:
\begin{enumerate}
    \item $\mathsf{vm\_id}() \to V_i$: returns the VM identifier.
    \item $\mathsf{execute}(\mathit{state}, \mathit{tx}) \to (\mathit{receipt}, \Delta)$: executes transaction $\mathit{tx}$ against state tree $\mathit{state}$, returning a receipt and a set of state changes $\Delta$.
\end{enumerate}
\end{definition}

The dispatcher maintains a registry $\mathcal{E}: \mathcal{V} \to E$ mapping each VM identifier to its engine.
Transaction execution follows Algorithm~\ref{alg:dispatch}.

\begin{algorithm}[H]
\caption{n-VM Transaction Dispatch}\label{alg:dispatch}
\begin{algorithmic}[1]
\Require Transaction $\mathit{tx}$ with payload $P$, state tree $S$
\Ensure Receipt $R$
\State $\mathit{opcode} \gets P[0]$
\State $V \gets \mathsf{Route}(\mathit{opcode})$
\If{$V = \bot$}
    \State \Return $\mathsf{Reject}(\text{``unknown opcode''})$
\EndIf
\State $E \gets \mathcal{E}[V]$
\If{$\neg\mathsf{Valid}(S, \mathit{tx})$}
    \State \Return $\mathsf{Reject}(\text{``invalid authorization, nonce, or fee''})$ \Comment{No state change}
\EndIf
\State $S' \gets \mathsf{Snapshot}(S)$ \Comment{Save state for rollback}
\State $(R, \Delta) \gets E.\mathsf{execute}(S, \mathit{tx})$
\If{$R.\mathit{success} = \mathtt{false}$}
    \State $\mathsf{Rollback}(S, S')$ \Comment{Restore pre-execution state}
\EndIf
\State $\mathsf{Apply}(S, \Phi(\mathit{tx}, R))$ \Comment{Fee debit and nonce increment, also on failure}
\State \Return $R$
\end{algorithmic}
\end{algorithm}

Here $\mathsf{Valid}(S, \mathit{tx})$ checks the authorization (attestation or raw-chain signature), that the transaction nonce equals the sender's current nonce in $S$ (for raw Solana and Tron ingress, that the signature or transaction ID is not yet recorded; for raw Bitcoin ingress, that every input is unspent), and that the sender's balance covers the maximum fee the transaction can be charged. A transaction that fails this check is excluded from the block and changes no state, so a transaction can be applied at most once.
The snapshot/rollback mechanism ensures that a failed transaction in any VM does not corrupt the shared state tree.

% ----------------------------------------------------------
\subsection{Block Execution}

A block $B = [\mathit{tx}_1, \mathit{tx}_2, \ldots, \mathit{tx}_m]$ is an ordered sequence of transactions that may target different VMs.
The dispatcher processes each transaction sequentially (in the baseline mode), applying state changes atomically:

\begin{equation}\label{eq:block}
    S_j = \begin{cases}
        S_{j-1} \oplus \Delta_j \oplus \Phi_j & \text{if } \mathit{tx}_j \text{ succeeds}, \\
        S_{j-1} \oplus \Phi_j & \text{otherwise},
    \end{cases}
    \qquad j = 1, \ldots, m.
\end{equation}

Here $S \oplus \Delta$ denotes the state $S$ with every key written by $\Delta$ set to its new value (a debit is a write of the reduced balance), and $\Phi_j = \Phi(\mathit{tx}_j, R_j)$ debits the sender's fee for the resources consumed and advances the sender's nonce (or, for raw Solana, Tron, and Bitcoin ingress, records the consumed signature, transaction ID, or inputs). It is applied whether or not execution succeeds, so a failed transaction still pays for the work it caused and its signature cannot be included again.
The final state $S_m$ and the ordered receipt list $[R_1, \ldots, R_m]$ are deterministic given the initial state $S_0$ and the block $B$.

% ============================================================
\section{Unified Identity Layer}\label{sec:identity}

% ----------------------------------------------------------
\subsection{Cross-VM Address Derivation}

The central identity primitive is the deterministic mapping from a single $\texttt{id\_com}$ to VM-specific addresses.

\begin{definition}[Address Derivation]\label{def:addr}
For a VM $V_i$ with native address length $\ell_i$ bytes, the derived address is:
\[
    \alpha_{V_i} = \mathsf{Truncate}_{\ell_i}\big(\SHA(\mathsf{tag}_i \| \texttt{id\_com})\big)
\]
where $\mathsf{tag}_i$ is a VM-specific domain separator string and $\mathsf{Truncate}_\ell$ takes the last $\ell$ bytes of the hash output.
\end{definition}

For the $n=5$ instantiation:

\begin{align}
    \alpha_{\text{EVM}} &= \SHA(\texttt{"evm:"} \| \texttt{id\_com})[12{:}32] & &\text{(20 bytes)} \label{eq:evm} \\
    \alpha_{\text{SVM}} &= \SHA(\texttt{"svm:"} \| \texttt{id\_com}) & &\text{(32 bytes)} \label{eq:svm} \\
    \alpha_{\text{BVM}} &= \SHA(\texttt{"bvm:"} \| \texttt{id\_com}) & &\text{(32 bytes)} \label{eq:bvm} \\
    \alpha_{\text{TVM}} &= \SHA(\texttt{"tron:"} \| \texttt{id\_com})[12{:}32] & &\text{(20 bytes)} \label{eq:tvm} \\
    \alpha_{\text{native}} &= \texttt{id\_com} & &\text{(32 bytes)} \label{eq:native}
\end{align}

\begin{theorem}[Address Isolation]\label{thm:isolation}
For any two distinct VM tags $\mathsf{tag}_i \ne \mathsf{tag}_j$ and any $\texttt{id\_com}$, the derived addresses $\alpha_{V_i}$ and $\alpha_{V_j}$ are independent uniformly distributed strings when SHA-256 is modeled as a random oracle.
\end{theorem}

\begin{proof}[Proof sketch]
SHA-256 is modeled as a random oracle.
Since $\mathsf{tag}_i \ne \mathsf{tag}_j$, the inputs $\mathsf{tag}_i \| \texttt{id\_com}$ and $\mathsf{tag}_j \| \texttt{id\_com}$ are distinct, so their outputs are independently distributed.
Truncation preserves computational independence: knowledge of a truncated output does not help predict a different truncation of a different hash.
\end{proof}

Isolation here means namespace separation: an address in one VM's namespace does not collide with, or predict, an address in another VM's namespace except with negligible probability. It does not give unlinkability. Because $\texttt{id\_com}$ is public on-chain and the derivation is public, anyone who knows $\texttt{id\_com}$ can compute all of its VM addresses and link them.

% ----------------------------------------------------------
\subsection{Reverse Resolution}

Since Equation~\eqref{eq:evm} truncates a hash, $\texttt{id\_com}$ cannot be computed from $\alpha_{\text{EVM}}$. For interoperability with legacy tooling (e.g., block explorers that display EVM addresses), the dispatcher therefore maintains a stored reverse index $\alpha_{\text{EVM}} \mapsto \texttt{id\_com}$, written when an account is created in the state tree; for an address with no entry, it falls back to the legacy identity $\SHA(\texttt{"legacy\_evm:"} \| \alpha_{\text{EVM}})$ of Table~\ref{tab:raw} if that account exists. A TRON index is kept in the same way.
Separately, the \texttt{cross\_vm\_call} precompile exposes the deterministic one-way mapping
\begin{equation}
    \mathsf{ace\_id\_from\_evm}(\alpha_{\text{EVM}}) = \SHA(\texttt{"ace\_from\_evm:"} \| \alpha_{\text{EVM}}),
\end{equation}
which gives an EVM address an ACE-side identifier without the reverse index. Its output is not the account's $\texttt{id\_com}$ and is not used to locate balances.

% ----------------------------------------------------------
\subsection{Raw Chain Ingress}\label{sec:raw-ingress}

So that users can start with the wallets they already have, the \NVM{} chain must accept transactions signed with legacy wallets (MetaMask for EVM, Phantom for Solana, etc.) that have no knowledge of the \ACEGF{} attestation mechanism.

\begin{definition}[Raw Chain Transaction]\label{def:raw}
A raw chain transaction is a transaction that carries a native chain signature (ECDSA for EVM, Ed25519 for Solana, etc.) in addition to the standard payload.
The dispatcher performs:
\begin{enumerate}
    \item \textbf{Signature verification}: recover the signer's public key or address using the chain-native verification algorithm.
    \item \textbf{Identity mapping}: compute the deterministic $\texttt{id\_com}$ for the recovered address.
    \item \textbf{Payload reconstruction}: verify that the transaction payload matches the canonical form derived from the raw bytes.
    \item \textbf{Account binding}: ensure the derived $\texttt{id\_com}$ has an account in the state tree, creating one if necessary.
\end{enumerate}
\end{definition}

This mechanism supports four ingress paths:

\begin{table}[H]
\centering
\caption{Raw chain ingress verification.}
\label{tab:raw}
\begin{tabular}{llll}
\toprule
\textbf{Chain} & \textbf{Signature} & \textbf{Address Format} & \textbf{Identity Derivation} \\
\midrule
EVM & ECDSA/secp256k1 & 20-byte Keccak & $\SHA(\texttt{"legacy\_evm:"} \| \mathit{addr})$ \\
Solana & Ed25519 & 32-byte pubkey & $\SHA(\texttt{"legacy\_sol:"} \| \mathit{pubkey})$ \\
Bitcoin & ECDSA/Schnorr & P2WPKH/Taproot & $\SHA(\texttt{"legacy\_btc:"} \| \mathit{pubkey})$ \\
Tron & ECDSA/secp256k1 & 20-byte (T-prefix) & $\SHA(\texttt{"legacy\_tron:"} \| \mathit{addr})$ \\
\bottomrule
\end{tabular}
\end{table}

With raw chain ingress, users can begin interacting with the \NVM{} chain using their existing wallets, then optionally upgrade to \ACEGF{} identities for cross-VM unification and installed-key authorization.
Upgrading moves balances from the legacy identity to the new $\texttt{id\_com}$ by an ordinary transfer signed with the legacy wallet; this paper defines no protocol that merges the two identifiers.

% ============================================================
\section{Unified Token Ledger}\label{sec:token}

% ----------------------------------------------------------
\subsection{Design}

Token standards differ across VMs.
EVM uses ERC-20 (balance-of mapping with allowances); SVM uses SPL tokens (separate token accounts with mint/owner metadata).
These are semantically equivalent but structurally incompatible.

The \NVM{} token ledger stores all token state in a single built-in program account within the state tree, and exposes VM-native interfaces as views over this shared storage:

\begin{definition}[Unified Token State]\label{def:token}
For a token with mint identifier $M$ (32 bytes), the canonical state consists of:
\begin{itemize}
    \item $\mathsf{supply}(M)$: total supply (uint64).
    \item $\mathsf{decimals}(M)$: decimal precision (uint8).
    \item $\mathsf{authority}(M)$: mint authority $\texttt{id\_com}$ (32 bytes).
    \item $\mathsf{balance}(M, \texttt{id\_com})$: balance for identity (uint64).
    \item $\mathsf{allowance}(M, \texttt{id\_com}_{\text{owner}}, \texttt{id\_com}_{\text{spender}})$: delegated allowance (uint64).
\end{itemize}
\end{definition}

Storage slots are derived via domain-separated hashing:
\begin{align}
    \mathsf{slot}_{\text{balance}} &= \SHA(\texttt{"balance:"} \| M \| \texttt{id\_com}) \\
    \mathsf{slot}_{\text{allowance}} &= \SHA(\texttt{"allowance:"} \| M \| \texttt{id\_com}_{\text{owner}} \| \texttt{id\_com}_{\text{spender}})
\end{align}

% ----------------------------------------------------------
\subsection{Multi-Interface Access}

The same underlying balance is accessible through both EVM and SVM interfaces:

\paragraph{ERC-20 interface.}
Each mint $M$ has a deterministic ERC-20 contract address:
\begin{equation}
    \mathit{addr}_{\text{ERC20}} = \SHA(\texttt{"erc20-addr:"} \| M)[12{:}32]
\end{equation}
When an EVM contract calls \texttt{balanceOf($\alpha_{\text{EVM}}$)} on this address, the runtime:
(1)~resolves $\alpha_{\text{EVM}}$ to the corresponding $\texttt{id\_com}$ via the reverse index,
(2)~reads $\mathsf{balance}(M, \texttt{id\_com})$ from the unified ledger, and
(3)~returns the result in ERC-20-compatible ABI encoding.

\paragraph{SPL interface.}
SPL token accounts are modeled as compatibility aliases.
An associated token address (ATA) is derived using the standard Solana PDA algorithm:
\begin{equation}
    \mathit{ata} = \mathsf{FindPDA}([\mathit{owner\_pubkey}, \mathit{spl\_program\_id}, M], \mathit{ata\_program\_id})
\end{equation}
When an SVM program queries a token account's balance, the runtime resolves the ATA to the owner's $\texttt{id\_com}$ and reads from the same unified ledger.

% ----------------------------------------------------------
\subsection{Cross-VM Transfer Atomicity}

\begin{theorem}[Cross-VM Atomicity]\label{thm:atomicity}
A token transfer from an EVM context to an SVM context executes as a single state transition on the unified ledger, requiring no intermediate bridge state.
\end{theorem}

\begin{proof}[Proof sketch]
Both the EVM and SVM interfaces resolve to operations on the same storage slots keyed by $\texttt{id\_com}$.
A transfer from $\texttt{id\_com}_A$ to $\texttt{id\_com}_B$ modifies exactly two slots:
\[
    \mathsf{balance}(M, \texttt{id\_com}_A) \mathrel{{-}{=}} v, \qquad
    \mathsf{balance}(M, \texttt{id\_com}_B) \mathrel{{+}{=}} v.
\]
Since both writes occur within a single transaction on a single state tree, they are atomic by the transaction execution semantics of Equation~\eqref{eq:block}.
No intermediate ``locked'' or ``in-flight'' state exists.
\end{proof}

A cross-VM transfer therefore needs no lock-mint-burn-release cycle, multi-signature committee, or optimistic challenge period.

% ============================================================
\section{Parallel Execution}\label{sec:parallel}

% ----------------------------------------------------------
\subsection{Read/Write-Set Conflict Detection}

To run independent transactions in parallel (e.g., EVM contract calls and native transfers targeting disjoint accounts), the \NVM{} scheduler extracts read and write sets from each transaction and builds conflict-free batches.

\begin{definition}[Read and Write Sets]\label{def:writeset}
The read set $R(\mathit{tx})$ and write set $W(\mathit{tx})$ of a transaction are the sets of state keys (accounts and storage slots) it may read and may write, including the keys written by $\Phi(\mathit{tx}, \cdot)$ (the sender's account and, for raw Solana, Tron, and Bitcoin ingress, the recorded signature, transaction ID, or inputs). Each set is:
\[
    X(\mathit{tx}) = \begin{cases}
        \{a_1, a_2, \ldots\} & \text{if the accessed keys can be statically determined}, \\
        \top \text{ (global)} & \text{otherwise (conflicts with all transactions),}
    \end{cases}
\]
for $X \in \{R, W\}$.
\end{definition}

For transactions whose access sets are statically determinable (native transfers, SVM transfers, BVM transfers), the scheduler extracts the sender and recipient account IDs directly from the payload; such a transfer reads and writes both balances, so $R(\mathit{tx}) = W(\mathit{tx})$.
For transactions with unbounded side effects (EVM \texttt{CALL}, EVM \texttt{CREATE}, SVM \texttt{INVOKE}), both sets are conservatively marked as global.

\begin{definition}[Conflict]\label{def:conflict}
Two transactions $\mathit{tx}_i$ and $\mathit{tx}_j$ conflict if one writes a key that the other reads or writes:
\[
    \begin{aligned}
    & W(\mathit{tx}_i) \cap \big(R(\mathit{tx}_j) \cup W(\mathit{tx}_j)\big) \ne \emptyset \\
    \lor\; & W(\mathit{tx}_j) \cap \big(R(\mathit{tx}_i) \cup W(\mathit{tx}_i)\big) \ne \emptyset,
    \end{aligned}
\]
where $\top$ intersects every non-empty set.
\end{definition}

\begin{algorithm}[H]
\caption{Batch Construction}\label{alg:batch}
\begin{algorithmic}[1]
\Require Ordered transaction list $[\mathit{tx}_1, \ldots, \mathit{tx}_m]$
\Ensure Ordered list of conflict-free batches $[B_1, B_2, \ldots]$
\State $\mathit{batches} \gets [\,]$; $\mathit{current} \gets [\,]$; $\mathit{used}_R \gets \emptyset$; $\mathit{used}_W \gets \emptyset$
\For{$j = 1, \ldots, m$}
    \State $r \gets R(\mathit{tx}_j)$; $w \gets W(\mathit{tx}_j)$
    \State $c \gets \big(w \cap (\mathit{used}_R \cup \mathit{used}_W) \ne \emptyset\big) \lor \big(r \cap \mathit{used}_W \ne \emptyset\big)$
    \If{$r = \top$ \textbf{or} $w = \top$ \textbf{or} $c$}
        \If{$|\mathit{current}| > 0$}
            \State Append $\mathit{current}$ to $\mathit{batches}$
        \EndIf
        \State Append $[\mathit{tx}_j]$ to $\mathit{batches}$ \Comment{Conflicting or global tx runs alone}
        \State $\mathit{current} \gets [\,]$; $\mathit{used}_R \gets \emptyset$; $\mathit{used}_W \gets \emptyset$
    \Else
        \State Append $\mathit{tx}_j$ to $\mathit{current}$
        \State $\mathit{used}_R \gets \mathit{used}_R \cup r$; $\mathit{used}_W \gets \mathit{used}_W \cup w$
    \EndIf
\EndFor
\If{$|\mathit{current}| > 0$}
    \State Append $\mathit{current}$ to $\mathit{batches}$
\EndIf
\State \Return $\mathit{batches}$
\end{algorithmic}
\end{algorithm}

Transactions within the same batch are executed in parallel (via Rayon~\cite{rayon} in the implementation); batches are executed sequentially to preserve block-level determinism.
No two transactions in a batch conflict (Definition~\ref{def:conflict}), so no transaction in a batch reads or writes a key that another writes, and every ordering of a batch yields the same state as its block order. The read sets matter here: a transaction whose outcome depends on a balance that another transaction in the same batch modifies would otherwise observe a value that differs from sequential execution.

% ----------------------------------------------------------
\subsection{Account-Keyed Sharding}

For further parallelism, the \NVM{} dispatcher partitions state by identity and routes each transaction by the state it touches.
The shard key cannot be a tag chosen by the sender: balances in the unified ledger are keyed by $\texttt{id\_com}$ and shared by all VMs, so two transactions that spend the same balance under different tags, or through different VMs, would land in different shards and both execute.

\begin{definition}[Shard Routing]\label{def:shard}
Given $K$ shards (default $K = 64$), every state key $a$ is owned by one identity $\mathsf{owner}(a)$ (an account is owned by itself; a ledger slot such as $\mathsf{balance}(M, \texttt{id\_com})$ is owned by $\texttt{id\_com}$; a contract's storage is owned by the contract account; a recorded signature, transaction ID, or spent input is owned by the identity it debits), and
\begin{equation}
    \sigma(a) = \SHA\big(\mathsf{owner}(a)\big) \bmod K.
\end{equation}
A transaction $\mathit{tx}$ is assigned to shard $s$ if $R(\mathit{tx}) \ne \top$, $W(\mathit{tx}) \ne \top$, and $\sigma(a) = s$ for every $a \in R(\mathit{tx}) \cup W(\mathit{tx})$. Otherwise it is assigned to the \emph{shared shard}.
\end{definition}

The VM prefix does not enter $\sigma$, so all spends from one balance through any VM map to the same shard. Any optional context tag a transaction carries does not affect its placement.
Cross-shard rule: the block is cut at every shared-shard transaction. Within a segment between two such cuts, the $K$ shards run in parallel and each executes its transactions in block order; each shared-shard transaction then runs alone, after the preceding segment and before the next. Transactions in different shards of a segment touch disjoint keys, so the result equals sequential execution of the block in its original order (Equation~\eqref{eq:block}), and two spends of the same balance are always serialized in one shard or across a cut.

% ----------------------------------------------------------
\subsection{Throughput Analysis}

We model throughput under three execution strategies.

\begin{table}[H]
\centering
\caption{Projected throughput under different execution strategies.}
\label{tab:throughput}
\begin{tabular}{llr}
\toprule
\textbf{Strategy} & \textbf{Description} & \textbf{TPS (projected)} \\
\midrule
Sequential & Single-threaded, all VMs & ${\sim}6{,}700$--$10{,}000$ \\
Parallel batching, all conflicting & Checks on 16 cores, serial execution & ${\sim}16{,}000$--$59{,}000$ \\
Account-keyed sharding, no conflicts & 64 shards on 16 cores & ${\sim}107{,}000$--$161{,}000$ \\
\bottomrule
\end{tabular}
\end{table}

The throughput projections assume:
\begin{itemize}
    \item 400\,ms block time, of which ${\sim}350$\,ms is available for checks and execution after ${\sim}50$\,ms of block publication (the \ACERuntime{} pipeline~\cite{aceruntime}).
    \item ${\sim}1$--$5\,\mu$s header check and ${\sim}76\,\mu$s Ed25519 signature verification per transaction, single core.
    \item ${\sim}10$--$50\,\mu$s execution per transaction (VM-dependent).
    \item 16-core commodity server hardware; header and signature checks are independent and spread over all cores.
    \item 70\% of transactions have statically determinable read and write sets (native and simple transfers).
\end{itemize}

Each row is $350\,\mathrm{ms}$ divided by the wall-clock time per transaction, scaled to one second. Sequential: $76+5+50=131\,\mu$s to $76+1+10=87\,\mu$s. All conflicting: $(76+5)/16+50\approx 55\,\mu$s to $(76+1)/16+10\approx 15\,\mu$s. No conflicts: $(76+5+50)/16\approx 8.2\,\mu$s to $(76+1+10)/16\approx 5.4\,\mu$s. The last row is the CPU ceiling of a 16-core server; it assumes high shard locality (most transactions are intra-shard) and minimal cross-shard traffic, and 64 shards do not raise it beyond 16 cores. With ML-DSA-44 keys (${\sim}200\,\mu$s per verification) the all-conflicting row becomes ${\sim}14{,}000$--$39{,}000$ TPS, and the no-conflict row is capped by the ${\sim}47{,}700$\,tx/s bandwidth ceiling of 1\,Gbps for 2{,}623-byte transactions~\cite{aceruntime}. These are projections from cited per-operation costs, not measurements.

% ============================================================
\section{Cross-VM Communication}\label{sec:crossvm}

% ----------------------------------------------------------
\subsection{EVM Precompiles}

To enable EVM smart contracts to interact with the \NVM{} identity and cross-VM layers, the runtime exposes a set of precompiled contracts at reserved addresses:

\begin{table}[H]
\centering
\caption{ACE EVM precompiles.}
\label{tab:precompiles}
\begin{tabular}{lll}
\toprule
\textbf{Address} & \textbf{Name} & \textbf{Function} \\
\midrule
\texttt{0x0100} & \texttt{id\_com\_verify} & Verify an identity commitment proof \\
\texttt{0x0101} & \texttt{context\_derive} & Derive context from inputs \\
\texttt{0x0102} & \texttt{admin\_factor\_check} & Check admin factor credentials \\
\texttt{0x0103} & \texttt{zkace\_batch\_verify} & Batch-verify ZK-ACE proofs \\
\texttt{0x0104} & \texttt{multisig\_derive} & Derive multisig address \\
\texttt{0x0105} & \texttt{multisig\_verify} & Verify multisig attestations \\
\texttt{0x0106} & \texttt{cross\_vm\_call} & Encode an EVM-originated cross-VM invocation request \\
\texttt{0x0107} & \texttt{resolve\_svm\_addr} & Resolve $\texttt{id\_com} \to$ SVM address \\
\bottomrule
\end{tabular}
\end{table}

The \texttt{cross\_vm\_call} precompile (0x0106) provides an EVM-facing interface for expressing an SVM-targeted invocation request by supplying the program ID and instruction data.
In the present design, this precompile defines the call envelope and dispatcher routing hook for cross-VM requests.
A fully general synchronous return-value path with shared gas metering is left as future work.

% ----------------------------------------------------------
\subsection{SVM Syscalls}

Symmetrically, SVM programs can access the identity layer and complementary cross-VM routing hooks through custom syscalls:

\begin{itemize}
    \item \texttt{ace\_attest}: verify an attestation credential within an SVM program.
    \item \texttt{ace\_id\_com}: query the $\texttt{id\_com}$ for the current transaction sender.
    \item \texttt{ace\_cross\_vm\_call}: express an EVM-targeted cross-VM invocation request from an SVM program.
\end{itemize}

With the EVM precompiles, these syscalls give EVM and SVM code access to the identity layer and a way to route cross-VM requests in both directions.
They specify how multi-VM interaction is encoded and dispatched, while a fully general synchronous cross-VM execution model remains future work.

% ============================================================
\section{Security Analysis}\label{sec:security}

% ----------------------------------------------------------
\subsection{VM Isolation}

\begin{proposition}[Execution Isolation]\label{prop:isolation}
A failure (revert, out-of-gas, panic) in VM $V_i$ during execution of transaction $\mathit{tx}$ does not affect the state of any other VM $V_j$ ($j \ne i$) or any subsequently executed transaction.
\end{proposition}

\begin{proof}[Proof sketch]
By Algorithm~\ref{alg:dispatch}, the dispatcher takes a state snapshot before delegating to any engine.
On failure, the state is rolled back to the snapshot.
Since the state tree is the sole shared resource and rollback restores it to the pre-execution state, the only effects of the failed transaction that persist are its fee debit and nonce increment $\Phi$ (Eq.~\eqref{eq:block}), which touch only the sender's own account.
\end{proof}

% ----------------------------------------------------------
\subsection{Identity Binding Correctness}

\begin{proposition}[Raw Chain Binding]\label{prop:binding}
For a raw chain transaction from chain $c$ with recovered address $\mathit{addr}$, the binding to $\texttt{id\_com} = \SHA(\texttt{"legacy\_}c\texttt{:"} \| \mathit{addr})$ is:
\begin{enumerate}
    \item \textbf{Deterministic}: the same address always maps to the same $\texttt{id\_com}$.
    \item \textbf{Collision-resistant}: two distinct addresses map to different $\texttt{id\_com}$ values with overwhelming probability.
    \item \textbf{Domain-separated}: addresses from different chains map to different $\texttt{id\_com}$ values even if the raw bytes are identical.
\end{enumerate}
\end{proposition}

\begin{proof}
Properties (1) and (2) follow directly from SHA-256 being a deterministic, collision-resistant hash function.
Property (3) follows from the chain-specific prefix: if $c_1 \ne c_2$, then $\texttt{"legacy\_}c_1\texttt{:"} \ne \texttt{"legacy\_}c_2\texttt{:"}$, so the inputs to SHA-256 are distinct even when $\mathit{addr}_1 = \mathit{addr}_2$.
\end{proof}

% ----------------------------------------------------------
\subsection{Deterministic Execution}

\begin{proposition}[Block-Level Determinism]\label{prop:determinism}
Given initial state $S_0$ and block $B$, the final state $S_m$ and receipt list $[R_1, \ldots, R_m]$ are identical across all validators.
\end{proposition}

\begin{proof}[Proof sketch]
The opcode routing function (Definition~\ref{def:routing}) is deterministic.
Each VM engine's \texttt{execute} function is deterministic (revm in Shanghai mode; SVM with fixed built-in programs; BVM with deterministic Script evaluation).
State snapshots and rollbacks are deterministic.
Parallel execution within batches produces the same result as sequential execution because no transaction in a batch writes a key that another transaction in the batch reads or writes (Definition~\ref{def:conflict}, Algorithm~\ref{alg:batch}).
Shard routing (Definition~\ref{def:shard}) is a deterministic function of the transaction's read and write sets, and under the cross-shard rule the sharded schedule produces the same result as sequential execution in block order.
\end{proof}

% ============================================================
\section{Implementation}\label{sec:implementation}

% ----------------------------------------------------------
\subsection{System Overview}

The \NVM{} dispatcher is implemented in Rust as the \texttt{ace-n-vm} crate, consisting of approximately 5,000 lines of Rust code plus 3,000 lines of tests.
Table~\ref{tab:engines} summarizes the five VM engines in the concrete instantiation.

\begin{table}[H]
\centering
\caption{Implemented VM engines ($n=5$).}
\label{tab:engines}
\begin{tabularx}{\textwidth}{llll>{\raggedright\arraybackslash}X}
\toprule
\textbf{VM} & \textbf{Opcodes} & \textbf{Backend} & \textbf{Address} & \textbf{Capabilities} \\
\midrule
Native & \texttt{0x01}--\texttt{0x0F} & \texttt{ace\_engine} & 32-byte $\texttt{id\_com}$ & Transfer, account creation, auth key management \\
EVM & \texttt{0x10}--\texttt{0x1F} & revm v19 (Shanghai) & 20-byte & Full EVM: Solidity, Vyper, ERC-20/721/1155/4626 \\
SVM & \texttt{0x20}--\texttt{0x2F} & Built-in programs & 32-byte & SPL tokens, SystemProgram, PDAs \\
BVM & \texttt{0x30}--\texttt{0x3F} & Script interpreter & 32-byte & Bitcoin Script, UTXO management, P2WPKH/Taproot \\
TVM & \texttt{0x40}--\texttt{0x4F} & revm (remapped) & 20-byte & Tron compatibility via opcode remapping \\
\bottomrule
\end{tabularx}
\end{table}

% ----------------------------------------------------------
\subsection{TVM as EVM Delegation}

The TVM engine has no separate execution engine; it reuses the EVM engine with opcode remapping:

\begin{equation}
    \mathsf{TVM.execute}(\mathit{tx}) = \mathsf{EVM.execute}\big(\mathsf{remap}(\mathit{tx})\big)
\end{equation}

where $\mathsf{remap}$ translates opcodes $\texttt{0x40} \to \texttt{0x10}$, $\texttt{0x41} \to \texttt{0x11}$, $\texttt{0x42} \to \texttt{0x12}$, and adjusts the address namespace from Tron to EVM.
This approach adds a new VM with minimal code while preserving a Tron-compatible execution model for the remapped opcode subset.

% ----------------------------------------------------------
\subsection{Raw Chain Verification}

Each raw chain ingress path implements chain-specific signature recovery and payload canonicalization:

\begin{itemize}
    \item \textbf{EVM}: recover the secp256k1 signer via \texttt{ecrecover}, verify chain ID matches the attestation domain, and reconstruct the canonical payload from the raw RLP-encoded transaction.
    \item \textbf{Solana}: verify the Ed25519 signature, extract the transfer parameters, and compute the sender's legacy $\texttt{id\_com}$.
    The transaction's \texttt{recent\_blockhash} must be the hash of one of the last $N$ \NVM{} blocks, and signatures are deduplicated across that window. A transaction signed for Solana mainnet, or one older than the window, is therefore rejected.
    \item \textbf{Bitcoin}: verify ECDSA or Schnorr signatures depending on the output type (P2PKH, P2WPKH, P2WSH, Taproot), reconstruct the canonical payload from the serialized transaction.
    Every input must be an unspent output in the \NVM{} BVM UTXO set.
    \item \textbf{Tron}: similar to EVM but with Tron-specific protobuf encoding and domain slot verification.
    Tron transactions carry no chain ID, so, as for Solana, \texttt{ref\_block\_hash} must match one of the last $N$ \NVM{} blocks, the \texttt{expiration} must not have passed, and transaction IDs are deduplicated across that window. A transaction signed for Tron mainnet is therefore rejected.
\end{itemize}

% ============================================================
\section{Related Work}\label{sec:related}

\paragraph{Multi-VM Layer-1 Chains.}
Sei v2~\cite{sei} combines EVM and CosmWasm execution within a Cosmos SDK chain, using pointer contracts to bridge state between VMs.
However, pointer contracts introduce an indirection layer that does not achieve unified identity: EVM and CosmWasm addresses remain distinct, and token transfers between VMs require explicit pointer interactions.
Artela~\cite{artela} extends EVM with WebAssembly extensions for custom runtime modules, but the WASM layer is subordinate to EVM rather than a co-equal execution environment.

\paragraph{Multi-VM Layer-2 Systems.}
Movement Labs~\cite{movement} deploys an EVM-compatible execution layer atop a Move-based settlement chain.
Eclipse~\cite{eclipse} runs Solana's SVM as an Ethereum Layer-2 rollup.
These approaches inherit the trust assumptions and latency of their respective L1/L2 settlement mechanisms, and identity remains fragmented across layers.

\paragraph{Cross-Chain Interoperability.}
Polkadot's XCM~\cite{polkadot} and Cosmos's IBC~\cite{ibc} enable message passing across independent chains but do not provide unified identity or shared state.
Each chain maintains its own account model, and token transfers require lock-mint or burn-release bridge mechanisms.

\paragraph{Parallel Execution.}
Solana's Sealevel~\cite{solana} pioneered account-level parallelism by requiring transactions to declare accessed accounts upfront.
Aptos's Block-STM~\cite{aptos} uses optimistic concurrency control with re-execution on conflict.
The \NVM{} scheduler combines read/write-set extraction with account-keyed sharding, targeting cross-VM parallelism rather than single-VM account-level parallelism.

% ============================================================
\section{Conclusion}\label{sec:conclusion}

We have presented \NVM{}, a generalized architecture for hosting $n$ heterogeneous virtual machines as co-equal execution engines on a single Layer-1 blockchain.
Opcode-based dispatch, unified identity via \ACEGF{} commitments, and a single-ledger token runtime with multi-interface access together eliminate the need for bridges, wrapped tokens, or multiple wallets when operating across different execution environments.

The architecture is instantiated with $n=5$ VMs (Native, EVM, SVM, BVM, TVM), but the framework is parametric in $n$: adding a new VM requires only implementing the engine interface and registering an opcode range, with no changes to the identity layer, token ledger, or existing engines.
The TVM engine, which reuses the existing EVM backend through opcode remapping and namespace translation, is one example.

Under the analytical assumptions of Section~\ref{sec:parallel}, read/write-set conflict detection and account-keyed sharding suggest a potential throughput range from ${\sim}16{,}000$ TPS when every transaction conflicts to a CPU ceiling of ${\sim}107{,}000$--$161{,}000$ TPS without conflicts on a 16-core server with Ed25519 installed keys.

Possible extensions include:
\begin{itemize}
    \item \textbf{WASM VM.} Adding a WebAssembly execution engine (e.g., for CosmWasm or ink! contracts) as VM$_6$ would extend compatibility to the Cosmos and Polkadot ecosystems.
    \item \textbf{Move VM.} Integrating the Move execution engine would provide compatibility with Aptos and Sui smart contracts.
    \item \textbf{Speculative parallel execution.} Replacing the conservative read/write-set scheduler with optimistic concurrency control (Block-STM style) could improve parallelism for EVM and SVM transactions whose read and write sets cannot be statically determined.
    \item \textbf{Cross-VM composability.} Generalizing the current precompile/syscall request hooks into synchronous cross-VM calls with shared gas metering would enable atomic multi-VM transactions.
\end{itemize}

% ============================================================
\bibliographystyle{plain}

\end{document}